\documentclass [11pt,letterpaper]{article}
\usepackage{amsthm,amssymb}
\usepackage[lined,algonl,boxed]{algorithm2e}
\usepackage{psfrag}
\usepackage{graphicx}
\usepackage{color}
\usepackage{appendix}
\usepackage{epsfig}
\newtheorem{lemma}{Lemma}
\newtheorem{theorem}{Theorem}

\newtheorem{definition}{Definition}
\newtheorem{observation}{Observation}

\newtheorem{fact}{Fact}
\makeatletter
\newcommand{\Rmnum}[1]{\expandafter\@slowromancap\romannumeral #1@}
\makeatother

\title{
	\textbf{Combinatorial Geometry of Graph Partitioning - \Rmnum{1}}
}
\author{	\Large{ Manjish Pal \footnote{A significant portion of this work was done when the author was a B-Tech, M-Tech dual degree
                                student at IIT-Kanpur, India} }\\\\
  \large{
               Princeton University }\\
	\large{{\tt mpal@cs.princeton.edu}}\\
}
\date{}

\begin{document}

\pagestyle{empty}

\maketitle

\thispagestyle{empty}

\begin{abstract}
The {\sc $c$-Balanced Separator} problem is a graph-partitioning problem in which given a graph 
$G$, one aims to find a cut of minimum size such that both the sides of the cut
have at least $cn$ vertices. In this paper, we present new directions of progress in the {\sc $c$-Balanced Separator} problem. 
More specifically, we propose a family of mathematical programs, 
that depend upon a parameter $p > 0$, and is an extension of the uniform version of the SDPs proposed
by Goemans and Linial for this problem. In fact for the case, when $p=1$, if one can solve this program
in polynomial time then simply using the Goemans-Williamson's randomized rounding algorithm for {\sc Max Cut} \cite{WG95} will give an
$O(1)$-factor approximation algorithm for  {\sc $c$-Balanced Separator} improving the best known approximation
factor of $O(\sqrt{\log n})$ due to Arora, Rao and Vazirani \cite{ARV}. This family of programs is not convex but one can transform them 
into so called \emph{\textbf{concave programs}} in which one optimizes a concave function over a convex feasible set. 
It is well known that the optima of such programs lie at one of the extreme points of the feasible set \cite{TTT85}.
Our main contribution is  a combinatorial characterization of some extreme points of the feasible set of the mathematical program, for $p=1$
case, which to the best of our knowledge is the first of its kind. We further demonstrate how this characterization can be
used to solve the program in a restricted setting. Non-convex programs have recently been investigated by Bhaskara and Vijayaraghvan \cite{BV11} 
in which they design algorithms for approximating Matrix $p$-norms although their algorithmic techniques are analytical in nature. 
It is important to note that the properties of concave programs allows one to apply techniques due to Hoffman \cite{H81} or Tuy \emph{et al} \cite{TTT85} 
to solve such problems with arbitrary accuracy that, for special forms of concave programs, converge in polynomial time.
\end{abstract}

\newpage

\section{Introduction}
Graph partitioning is a problem of fundamental importance both in practice and theory. Many problems belonging 
to the several areas of computer science namely clustering, PRAM emulation, VLSI layout, packet routing in networks
can be modeled as partitioning a graph into two or more parts ensuring that the number of edges in the cut is ``small''. The word ``small'' doesn't refer to finding the min-cut in the graph as it doesn't ensure that the number of vertices in both sides of the cut is large. To enforce this balance condition one needs to normalize the cut-size in some sense. For the known notions of normalization like \emph{conductance}, \emph{expansion} and \emph{sparsity}, finding optimal separators is NP-hard for general graphs. Hence, the objective is to look for efficient approximation algorithms.
Because of the huge amount of work done to design good approximation algorithm for these problems, 
graph partitioning has become one of the central objects of study in the theory of geometric 
embeddings and random walks. Two fundamental problems which we will focus on 
are {\sc Sparsest Cut} and {\sc Balanced Separator}. These graph partitioning problems originally
came up in the context of multi-commodity flows in which we are given a graph with capacities
on the edges and a set of pairs of vertices (also called source-destination pairs) each having
a demand and the aim is to find a cut that minimizes the ratio of capacity of the cut and the total demand
through the cut. When the demand and capacities are all unit then the problem is called \emph{uniform}
and in case of general demands and capacities the problem is called \emph{non-uniform}. 

\subsection{Uniform Version}
The first approximation algorithm for such graph partitioning problems,
came out of the study of Reimannian manifolds in form of the well known Cheegar's Inequality \cite{C70}
which says that if $\Phi(G)$ is the conductance of the graph and $\lambda$ is the second largest 
eigenvalue of graph Laplacian then $ 2\Phi(G) \geq \lambda \geq \Phi(G)^2/2$. 
Because of the quadratic factor in the lower bound, the true approximation is $\frac{1}{\Phi(G)}$ which 
in worst case can be $\Omega(n)$ in worst case. The first true approximation algorithm for {\sc Sparsest Cut}
and {\sc Graph Conductance} was designed by Leighton and Rao \cite{LR99} whose approximation factor was $O(\log n)$. 
This also gave an $O(\log n)$  \textbf{\emph{pseudo-approximation algorithm}} for {\sc $c$-Balanced Separator}. This
algorithm is referred to as a pseudo-approximation algorithm because 
instead of returning a $c$-balanced cut, it returns a $c'$-balanced cut for some fixed $c'<c$ whose 
expansion is at most $O(\log n)$ times the optimum expansion of best $c$-balanced cut. Their algorithm
was based on an LP framework motivated from the idea of Multi-commodity flows. Their main contribution
was to derive an approximate max-flow, min-cut theorem corresponding to multi-commodity flow problem 
and the sparsest cut. Subsequently, a number of results
were discovered which showed that good approximation algorithms exist when one is considering
extreme cases such as the number of edges in the graphs is either very small or very large. In fact,
it is known that for planar graphs one can find balanced cuts which are twice as optimal \cite{GSV94} and for 
graph with an average degree of $\Omega (n)$, one can design $(1+\epsilon)$-factor approximation algorithms
where $\epsilon > 0$ with running time polynomial in input size \cite{AKK} (such an algorithm is called a 
\emph{Polynomial Time Approximation Scheme} or PTAS). The approximation factor
of $O(\log n)$ was improved to $O(\sqrt{\log n})$ in a breakthrough paper by Arora, Rao and Vazirani \cite{ARV}.
Their algorithm is based on semi-definite relaxations of these problems . 
The techniques and geometric structure theorems proved in their paper have subsequently led to breakthroughs
in the field of metric embeddings. The basic philosophy behind these approximation algorithms is to embed the
vertices of the input graph in an abstract space and derive a ``nice'' cut in this space. Recently, following a series of papers 
graph expansion has been related to the {\sc Unique Games} that ultimately led to sub-exponential time algorithms for 
{\sc Unique Games} \cite{ABS10}.

\subsection{Non-uniform Version}
The non-uniform version of the cut problems is inextricably linked with low distortion metric embedding.
It is easy to see that cut problems can be framed as optimization over $l_1$ metric which in general is NP-Hard.
So the incentive is to embed the points in a space on which one can optimize efficiently for eg. the $l_2^2$ metric
over which can optimize using SDPs. More specifically, using ideas from {\sf ARV} and the measured descent technique
of \cite{KLMN05}, firstly Lee \cite{L05} gave an $O(\log n)$ approximation algorithm for the non-uniform {\sc Sparsest Cut},
which was later improved to  $O(\log ^{3/4} n)$ by Chawla, Gupta and Rache \cite {CGR08}. A major breakthrough came from Arora, Lee
and Naor \cite{ALN08} who improved this bound to  $O(\sqrt{\log n} \log \log n)$ almost matching an old lower bound due to Enflo \cite{E69}
which says that there is an $n$ point metric in $l_1$ which need $\Omega (\sqrt{\log n})$ distortion to be embedded into
$l_2$.  

\subsection{Negative Results}
Graph partitioning problems like {\sc Sparsest Cut} and {\sc Balanced Separator}
are considered to among the few NP-hard problems which have resisted various attempts to prove inapproximability 
results. After the result of {\sf ARV}, there has been a lot of impetus towards proving lower bounds on approximation factors. It has been 
shown by Ambuhl et al \cite{AMS07} that {\sc Sparsest Cut} can't have a PTAS unless NP-complete problems
can be solved in randomized sub-exponential time. Because of the strong connections 
between semi-definite programming and the {\sc Unique Games Conjecture (UGC)} of Khot \cite{Kh02}, inapproximability
results are also known which assume UGC. More specifically, in a breakthrough result, Khot and Vishnoi \cite{KV05} showed that 
UGC implies super-constant lower bounds on the approximation factor for the non-uniform version of the problems. 
Lee and Naor \cite{LN} gave an analytical proof of the result that by exhibiting an $n$ point metric on the Heisenberg Group 
that is of negative type and needs $\omega(1)$ distortion to be embedded in $l_1$. Recently, it has been shown by 
Cheeger, Kleiner and Naor \cite{CKN09}, that the integrality gap of the non-uniform version of the sparsest cut SDP is $\Omega (\log^{O(1)} n)$. 
Devanur et al \cite{DKSV06} showed that the integrality gap of the SDP relaxation of Arora-Rao-Vazirani is $\Omega(\log \log n)$ 
thereby disproving the original conjecture of ARV that the integrality gap of their SDP relaxation is atmost a constant.   

\subsection{Non-Convex Programming}
In this paper we work with a form of non-convex programs called
\emph{\textbf{Concave Programming}}.  In order to define concave programming one first needs to define a concave
function.  A function $f:\mathbb{R}^d \rightarrow \mathbb{R}$ with domain \textbf{dom}$f$ is said to be concave if 
\textbf{dom}$f$ is convex and for all $x,y \in$ \textbf{dom}$f$, $f(\lambda x + (1-\lambda) y) \geq \lambda f(x) + (1-\lambda)f(y)$ for all $\lambda \in [0-1]$. 
Therefore, $f$ is concave iff $-f$ is a convex function. 
Based on this definition one defines concave programming as a form of mathematical programming 
in which one optimizes a concave function over a convex feasible set. 
More formally, a concave programming problem can be written as
$ \left[ \min_{x \in C} f(x) \right ]$
where $C$ is a convex set in $\mathbb{R}^d$ and $f$ is a concave function. 
The following is well known result for concave programming \cite{H76}.
\begin{fact}
For every concave programming problem there is an extreme point
of the convex feasible set $C$ which globally minimizes the optimization problem.
\end{fact}
The first algorithm for concave programming was designed by Tuy \cite{T64} in a restricted scenario
when the feasible set is a polytope. A more general case, when the feasible set is convex
but not necessarily polyhedral, was solved by Horst \cite{H76} and subsequently by Hoffman \cite{H81}, Tuy and
Thai \cite{TT81}. General concave programming is NP-hard as $\{0,1\}$-integer programming can
be cast as a concave program. There has been work towards designing efficient algorithms 
for some special class of concave programming. A comprehensive list of works done in concave programming 
can be found in Vaserstein's homepage \cite{V}. Recently, using analytical techniques Bhaskara and
Vijayaraghvan \cite{BV11} have successfully used non-convex programs to design algorithms
for approximating matrix $p-$norms. \\
 
\subsection {Our Contributions }
Our main contribution is to initiate the study of combinatorial geometric properties of a
non-convex relaxation for the {\c- Balanced Separator} problem. We show that
an efficient solution to our proposed program will imply improved an approximation algorithm
for this problem. In section 2, we formally introduce the notions of sparsity and balanced cuts and
sketch the Semi-Definite relaxation for $c$-{\sc Balanced Separator}
of {\sf ARV}. We then start section 4 by introducing a family of relaxations for $c$-{\sc Balanced Separator} 
which is generated by a parameter $p > 0$ and show that one can use its solution to design
an $O(1)$	-factor approximation algorithm for the problem.  Our result, although conditional, proposes new directions of progress on this problem 
and also a family of optimization problems which are more powerful than semi-definite programs 
in the context of approximation algorithms. Section 5 and 6 are devoted to find interesting properties on the geometry of the
feasible region of our program and in section 7 we show how these properties can be used to design an
efficient algorithm to search over a subset of extreme point called \emph {vertices}. We end the paper with Section 8
in which we present conclusions and future directions.
 
\section {Problem Definition}
We now formally  define the versions of balanced graph partitioning problem that we
focus on, in this paper. 
\textbf{ $c$-{\sc Balanced Separator}} \footnote{In \cite{ARV} $c$-{\sc Balanced Separator} is defined as 
the minimum sparsity of $c$-balanced cuts, we will be working with a definition which upto constant factors 
is equivalent to their definition} \\
Given a graph $G = (V,E)$ with $|V| = n, |E| = m$, the $c$-{\sc Balanced Separator} problem is to find $\alpha_c(G)$ where 
$\displaystyle \alpha_c(G) = \min_{S \subset V, cn < |S| < (1-c)n} E(S,\bar{S})$. \\
Although out techniques can potentially be generalized to {\sc Sparsest Cut} \footnote{Given a graph $G = (V,E)$ with $|V| = n, |E| = m$, for each cut $(S, \bar{S})$ define \emph{sparsity} of the cut to be the quantity $A(S) = \frac{|E(S, \bar{S})|}{|S||\bar{S}|}$. 
The uniform sparsest cut problem is to find $\alpha(G)$ where \\
$\displaystyle \alpha(G) = \min_{S \subset V} A(S)$.} and other balanced graph partitioning problems.

\subsection{SDP Relaxation for {\sc $c$-Balanced Separator}}
Unifying the spectral and the metric based (linear programming) approaches, {\sf ARV} used the following SDP
relaxation to get an improved (pseudo)-approximation algorithm
for the $c$-{\sc Balanced Separator}. Let us call this program $SDP_{BS}$,

\begin{center}
$\displaystyle \min \frac{1}{4}\sum_{i,j \in E} \|v_i - v_j\| ^2 $\\
$\displaystyle \|v_i\|^2 = 1 \quad \quad \forall i $\\
$\displaystyle \|v_i - v_j\|^2 + \|v_j - v_k\| ^2 \geq \|v_i - v_k\| ^2  \quad \quad \forall i,j,k$ \\
$\displaystyle \sum_{i<j} \|v_i - v_j\| ^2 \geq 4c(1-c)n^2$
\end{center}

It is easy to see that this indeed is a \emph{vector program} (and hence an SDP)
and is a relaxation for the $c$-{\sc Balanced Separator} 
problem. To show that this is a relaxation we have to show that for every cut we can get an assignment of vectors
such that all the constraints are satisfied and the value of the objective function
is the size of the cut. Given a cut $(S,\bar{S})$ if one maps all the vertices in 
$S$ to a unit vector $\textbf{n}$ and the vertices in $\bar{S}$ to $-\textbf{n}$ then the value
of the function is indeed the cardinality of $E(S,\bar{S})$. 
The main idea behind their algorithm is to show that for any set of vectors which satisfy the constraints
of the SDP there always exist two disjoint subsets of ``large'' size such that for any two points belonging to 
different subsets the squared Euclidean distance between them is at least $\Omega\left(\frac{1}{\sqrt{\log n}}\right)$.
The same idea is also used to get an improved approximation algorithm for {\sc Sparsest Cut} in \cite{ARV}.
Subsequently, this key idea has crucially been used in various other SDP based approximation algorithms and in 
solving problems related to metric embeddings.  

\section{Non-Convex Relaxation for $c$-{\sc Balanced Separator}}
Consider the following family of optimization problems which depend on a parameter $p \geq 0$.
This family is essentially an extension of the semi-definite program proposed by {\sf ARV}.
Throughout the paper we will use $\|.\|$ to represent the $l_2$ norm. Let us call this 
family of programs $F_{BS}^p$.

\begin{center}
$\displaystyle \min \frac{1}{2^p}\sum_{i,j \in E} \|v_i - v_j\| ^p $\\
$\displaystyle \|v_i\|^2 = 1 \quad \quad \forall i $\\
$\displaystyle \|v_i - v_j\|^p + \|v_j - v_k\| ^p \geq \|v_i - v_k\| ^p  \quad \quad \forall i,j,k$ \\
$\displaystyle \sum_{i,j \in E} \|v_i - v_j\| ^2 \geq 4c(1-c)n^2$
\end{center}

Note that for $p = 2$ this is the SDP relaxation used by {\sf ARV}. For $p = 1$,               
we are mapping the points onto a unit sphere, therefore
we do not have to force the additional triangle inequality constraint of $l_2$ metric. 
The same mapping described for $SDP_{BS}$ of the vertices of the graph
onto the unit sphere allows us to conclude that each program in this family 
is also a relaxation for {\sc $c$-Balanced Separator}. In most part of the 
paper we will be working with the case $p=1$. Now it is easy to see that if we can
solve this program for $p=1$, then simply using the randomized rounding algorithm 
of Goemans and Williamson \cite{WG95} will give an $O(1)$-approximation algorithm for the problem.
This is because of the fact that the last constraint ensures that a random hyperplane will 
find two sets of large size on both sides with constant probability \cite{ARV}. Another way to
look at it is that in this program we are actually embedding the points in an $l_2$ metric
which in turn is in $l_1$ metric.Therefore we have the following theorem:

\begin{theorem}
An efficient algorithm for solving $F_{BS}^p$ for $p=1$ implies an $O(1)$-factor approximation
algorithm for {\sc c-Balanced Separator}. 
\end{theorem}

\section{A Concave Programming Formulation}
In this section, we consider the family of optimization problems $F_{BS}^{p}$ proposed above 
and transform it into a concave program. This formulation allows us to use the algorithms which have been developed
to solve a concave program with arbitrary accuracy. 
We now write $F_{BS}^{p}$ as a program with variables as matrix entries and not as $d$-dimensional
vectors. The variables in the new program are of the form $x_{ij} = \left \langle v_i,v_j \right \rangle$. Since all $v_i$'s
are unit vectors we can write $\|v_i - v_j\|$ as $\sqrt{2 - 2\left \langle v_i,v_j \right \rangle}$.
If we consider the matrix $X$ with $ij^{th}$ entry as $x_{ij}$ use the transformation $z_{ij} = (1 - x_{ij})$, the new
optimization problem becomes:

\begin{center}
$\displaystyle \min \frac{1}{2^{p/2}} \sum_{i,j \in E} z_{ij}^{p/2} $\\
$\displaystyle z_{ij}^{p/2} + z_{jk}^{p/2} \geq z_{ik}^{p/2}  \quad \quad \forall i,j,k$ \\
$\displaystyle \sum_{i < j} z_{ij} \geq c(1-c)n^2$ \\
$\displaystyle z_{ii} = 0 \quad \quad \forall i $\\
$\displaystyle \textbf{1} - Z \succeq 0$
\end{center}
where \textbf{1} is the matrix with all entries as 1. \\
Let us call the above program $\tilde{F}_{BS}^{p}$.
This formulation allows us to prove the following lemma:

\begin{theorem}\label{concave}
$\tilde{F}_{BS}^{p}$ is a concave program for $0<p<2$. 
\end{theorem}
\begin{proof}
See Appendix.
\end{proof}

\section{Case $p = 1$}
In this case the our feasibility problem now looks like the following:

\begin{center}
$\displaystyle \min \frac{1}{\sqrt {2}} \sum_{i,j \in E} \sqrt{z_{ij}} $\\
$\displaystyle \sqrt{z_{ij}} + \sqrt{z_{jk}} \geq \sqrt{z_{ik}}  \quad \quad \forall i,j,k$ \\
$\displaystyle \sum_{i < j} z_{ij} \geq c(1-c)n^2$ \\
$\displaystyle z_{ii} = 0 \quad \quad \forall i $\\
$\displaystyle \textbf{1} - Z \succeq 0$
\end{center}
where \textbf{1} is the matrix with all entries as 1. Since $-1 \leq x_{ij} \leq 1$, $0 \leq z_{ij} \leq 2$. 
Let us denote the region in $\mathbb{R}^d$, which satisfies the last two constraints as $\cal P$,
the triangle inequality constraints as $\cal T$ and the ``well-separated'' constraint as $\cal H$. We 
will denote by $\cal F$ the feasible region.
Also if $C$ is an inequality constraint, then $C^{*}$ be the equality constraint corresponding to it.
The following is an easy observation which follows essentially from the definition.

\begin{observation}
${\cal P} \subseteq {\cal T}$.
\end{observation}

\subsection{The 3-Dimensional Intuition}

\begin{figure}[h!]\label{3D}
 \begin{center}
  \psfrag{p1}{$p_1$} \psfrag{p2}{$p_2$} \psfrag{p3}{$p_3$} \psfrag{p4}{$p_4$} \psfrag{p5}{$p_5$} 
   \psfrag{p6}{$p_6$} \psfrag{p7}{$p_7$}
  \includegraphics[scale = 0.7]{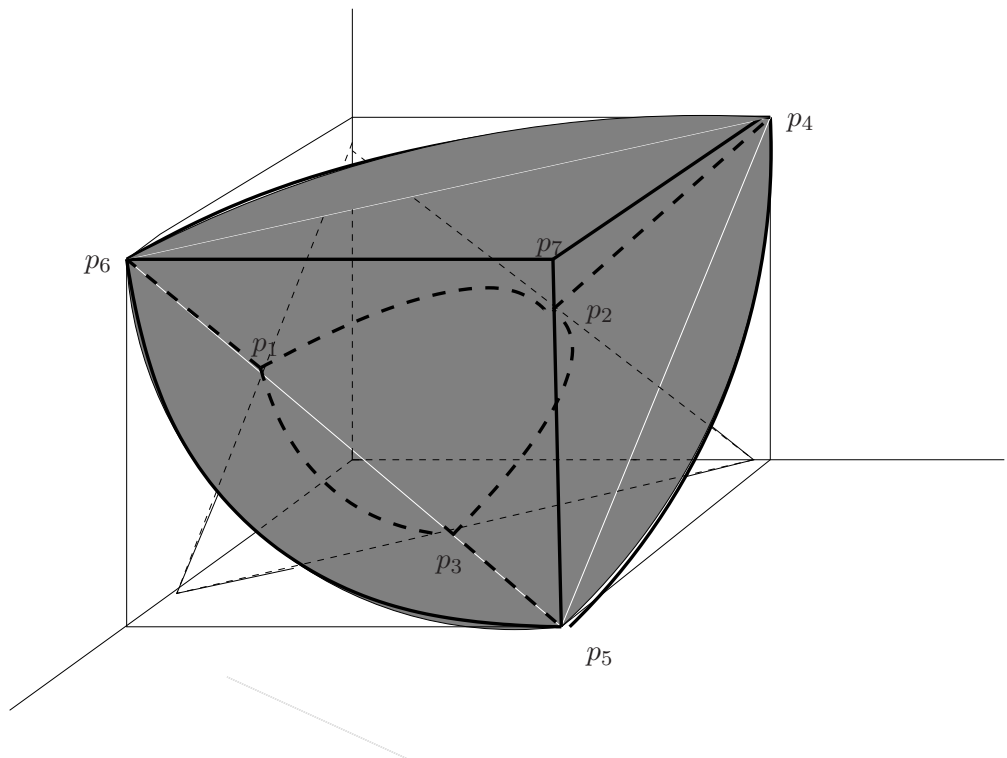} 
  \caption{The picture in 3-D}
 \end{center}
 \end{figure} 

If we just focus our attention to 3-variables and look at the feasible set (with out the positive semi-definite constraint),
then because of the nature of the triangle inequality constraints 
the geometry of the feasible set looks as shown in the Figure \ref{3D}. Notice that the feasible set is not
polyhedral but it has one dimensional line segments on its boundary. The shaded region enclosed by the 
points $p_1,p_2,p_3,p_4,p_5, p_6$ and $p_7$, depicts the feasible set. The line segments which are on the boundary of this object are
$p_1p_6$, $p_3p_5$, $p_2p_4 $. $p_6p_7$, $p_7p_4$ and $p_5p_7$. Also $p_1p_2$, $p_2p_3$ and $p_3p_1$ 
are non-linear arcs which are on the boundary of the feasible set.
 Given this description the following is easy to show. 

\begin{lemma}
Let $f = x_1^{p} + x_2^{p} + x_3^{p} $ where $p<1/2$ be the objective function to be minimized over
$\cal F$, then the optimum is achieved at one of extreme points $p_1$, $p_2$ or $p_3$.
\end{lemma}

\section{Combinatorial Geometry of the Feasible Set}
In this section we separately consider the constraints and develop tools
to understand the geometry of the feasible set which can potentially help us 
in getting an efficient algorithm to solve the feasibility problem. \\
Our aim in the sections to follow is to give a tight characterization of the 
``vertices''  of the proposed program which are defined
as follows:

\begin{definition}[Vertex] 
A point $p \in \mathbb{R}^d$ is called a vertex-set of the feasible set $\cal F$ if $p \in {\cal F}$ and there exists 
equality constraints $C_{1}^{*}, C_{2}^{*}, \dots C_{r}^{*} $ such that $p = \bigcap_k^{r} C_{k}^{*}$.
\end{definition}
 
\begin{definition}[Arc] 
 An arc $a$ of $\cal F$ is a closed one-dimensional curve joining two vertices of $\cal F$ such that there exists equality
 constraints $C_1^{*}, C_2^{*}, \dots C_r^{*} $ such that $a = \bigcap_i^{r} C_i^{*} $

\end{definition}
Notice that we can analogously define vertices and arcs corresponding to the 
regions ${\cal P}, {\cal H}$ and ${\cal T}$.
From the previous section it is clear that some arcs are line segments while others are not. The arcs 
which are line segments we will call them \emph {edges}. In the subsequent sections we will consider the
constraints separately.
 
\subsection{The Triangle Inequality Constraints}
We will now look at the geometric structure posed by the Triangle Inequality 
Constraints inside the $[0-2]^d$ hypercube and prove some structural results on those. 

\begin{definition}
Let $\cal R$ denote the region inside $[0-2]^d$ that is formed by the intersections of the
constraints $z_{ij} + z_{jk} \geq z_{kl}$, for all $i,j,k$.
\end{definition}

\begin{observation}
The $0$ vector is a vertex of $\cal T$.
\end{observation}

\begin{observation}
$e$(p) is an edge(vertex) of $\cal T$ iff it is an edge(vertex) of $\cal R$.
\end{observation}

We now take a deviation and first characterize all the symmetric $n \times n$ 
matrices with $0,1$ entries and main diagonal as $0$ which satisfy the 
triangle inequality constraints. Since such matrices represent the adjacency 
matrix of some graph, we essentially need to characterize all graphs whose corresponding 
matrices satisfy the triangle inequality constraints. As we will see later
this characterization will be helpful in analyzing the geometry of the triangle 
inequality constraints. But before that, we first need 
to define the following class of graphs:\\

\begin{definition}
A graph $G (V,E)$ is called \textbf{partial-clique} if there exists pairwise disjoint sets 
$S_1,S_2, \dots S_r \subseteq V$, such that $G = K_V \setminus \{\bigcup_{i=1}^r K_{S_i} \}$ where
$K_S$ denotes edges in the complete graph on $S \subseteq V$.
\end{definition}

We also define the following which will be of our interest later and subsequently
prove a series of combinatorial results based on these.

\begin{definition}
A partial-clique $G = K_V \setminus \{\bigcup_{i=1}^r K_{S_i} \} $ over a vertex set $V$ 
is called a \textbf{multi-clique} if $\bigcup_{i=1}^r S_i = V$.
\end{definition}
(In the literature these graphs are popularly known as multi-partite graphs.)
For the sake of brevity we will assume that the empty graph ($V = \emptyset$) 
is a multi-clique. We will also allow, again for the sake of simplification,
that one can take $S_i$'s of cardinality 1. 

\begin{theorem} \label{partialclique}
A graph satisfies the triangle inequality constraints if and only if it is a partial-clique. 
\end{theorem}
\begin{proof}
It is easy to see that if the graph is a partial clique then it satisfies the 
triangle inequality constraints. Let $S_1, S_2, \dots, S_r$ be the corresponding set 
of subsets. Consider any three vertices $v_i,v_j,v_k$, the following cases may arise: 
(i) none of of them lies in a subset $S$, (ii) all three lie in distinct $S_{i_1}, S_{i_2}, S_{i_3}$
(iii) two of them lie in a set $S_{i_1}$ and one lie in $S_{i_2}$  (iv) only one of them 
lies in a subset $S_{i_1}$ (v) all lie in the same subset $S_i$. In all these cases one can easily 
verify that the triangle inequality holds. Hence a partial-clique satisfies the triangle inequality
constraints. \\
For the converse part consider a graph that satisfies the triangle inequality constraints and assume
that it is not a partial-clique. Clearly the graph can't be disconnected because in that case one can 
easily find three vertices violating the corresponding triangle inequality. The triangle inequality 
essentially says that for all $i,j,k$, if two edges $v_iv_j$ and $v_jv_k$ are not present in the graph 
then the edge $v_iv_k$ should also not be present. 
Now assume that the complement of the graph has $m$ connected components.
Consider any one component say $H$. Let $S_1 \subseteq V$ be the set of vertices 
in the component $H$ with an edge $v_iv_j \in H$. If $H$ has just these two vertices 
$v_{i_1}$ and $v_{i_2}$ then it is already a clique. Let $v_{i_3}$ be a vertex which share an 
edge with $v_{i_1}$ or $v_{i_2}$. W.l.g let it be $v_{i_1}$.  Since $v_{i_1}v_{i_2}$ and $v_{i_1}v_{i_2}$ 
is in the complement, $v_{i_3}v_{i_2}$ also have to be in the complement, which forms a 3-clique. 
We can repeat the same argument for the next vertex $v_{i_4}$, which will share an edge with at least
one vertex in $\{v_{i_1},v_{i_2},v_{i_3} \}$, to show that it forms a 4-clique. Repeating this
argument for all the vertices of the component we can show that $H$ is a clique. The same 
holds for all the components.  
\end{proof}

\begin{lemma}\label{triangleedge}
The edges of $\cal R$ are of the form $\lambda B$ where $B$ is a bi-clique 
on $V = \{1,2, \dots, n\}$.
\end{lemma}
\begin{proof}
 One way is easy to verify. For the other side,
Notice that if the intersection of a set of equalities actually results
into a line then there will be a set of variables $x_{ij}$ such that all of them are
equal and rest are all zeros (hence there is just one variable). Therefore
the line will actually be a vector with some entries as $\lambda$ and rest as 0.
Since the 0/1 vectors which satisfy the triangle inequalities are the partial-cliques
such a vector with pass through a partial clique $G_v$. Consider the graph $G_{\lambda}v$
which represents a weighted partial-clique with all edges with weight $\lambda$. 
Now if such a partial clique is not a bi-clique then either one of the two cases are
possible: 
1. One can find an edge $\{i,j\}$ such that weight of ${\{i,j\}}$ is $\lambda$ and there is no
pair of the form $\{j,k\}$ or $\{i,k'\}$ such that $w_{\{j,k\}} = 0$ or $w_{\{i,k'\}} = 0$. \\
2. $G_\lambda$ is a multi-clique with $\lambda$ as edge weights. \\

If it is the first case let $\{i,j\}$ and $\{i_1,j_1\}$ be the two pairs which have weight 
$\lambda$ and since both have value $\lambda$, the intersection 
of the planes which we have chosen implies $z_{ij} = z_{i_1j_1}$. This means there must be some 
be some $k$ such that one of $z_{ij} + z_{jk} = z_{ik}, z_{ij} + z_{ik} = z_{jk} $ and
$z_{ik} + z_{jk} = z_{ij}$ is chosen and $z_{jk} = z_{i_1j_1}$ or $z_{ik} = z_{i_1j_1}$ is implied
by the rest of the planes chosen. But in both these cases, one of $z_{ik}$ or $z_{jk}$ is zero 
which is a contradiction.  \\
For the second case, let $i,j,k$ be three vertices lying in sets $V_1,V_2$ and $V_3$,
whose cliques are removed, respectively. Now due to the way equality of two 
variables is implied by a set of constraints, for the edge $\{i,j\}$ to have the same 
value as $\{i,k\}$ there must exist some $j' \in V_2$ and $k' \in V_3$ such that 
the plane containing variable $x_{ij'}$ and $x_{ik'}$ is chosen, buth this implies 
$x_{j'k'} = 0$ which is a contradiction.
\end{proof}

\begin{lemma}\label{triangle}
Let $T$ be the set of all $0/1$ vectors in the hypercube which satisfy the triangle inequality 
constraints, then $T$ is exactly the set of all $0$-dimensional faces of $\cal R$.
\end{lemma}
\begin{proof}
It is easy to see that the vertices of the cube are $n$-dimensional 
$0/1$ vectors and edges are formed by joining those vertices which have hamming distance 1. 
Since $\cal R$ is essentially the intersection of the Hamming cube with 
the unbounded polytope corresponding to the triangle inequalities, all the vertices of the
cube which satisfy the triangle inequality constraints will also be the vertices of $\cal R$.
We only need to show that there is no other vertices of $\cal R$. Since vertices form the 
boundary of the edges, the vertices of $\cal R$ are formed as a result of the intersection of edges
of the cube $\cal H$ with $\cal {P}$ or edges of $\cal {P}$ with $\cal {H}$. But we can show that every edge of $\cal H$ intersects the boundary of $\cal {P}$ only at its end points and also 
that every edge of $\cal {P}$ intersects the supporting planes of $\cal {H}$ only
at the vertices of $\cal {H}$ hence the vertices of $\cal R$
can only be the vertices of $\cal {H}$. To show the first claim let $(a_{12}, a_{13}, \dots, \lambda, \dots, a_{(n-1)n}, a_{nn})$
be an edge of $\cal H$ in which all $a_{ij}$'s except one are fixed to either 0 or 1 and
only one coordinate is varying as $\lambda \in [0-1]$. Now consider any plane corresponding to
the triangle inequality constraints of the form $x_{ij} + x_{jk} = x_{ik}$. Since there is 
only one co-ordinate in the line this equality can't be satisfied for any $0 < \lambda < 1$,
and hence intersection is only possible when $\lambda$ is either 0 or 1. Based on
Lemma \ref{triangle-edge} it is easy to see that the edges intersect the supporting
planes of $\cal {H}$ only at vertices of $\cal {H}$.
\end{proof}

\begin{theorem}\label{edges}
The line segment joining two vertices $u$ and $v$ of $\cal R$ is an edge of $\cal R$ 
if and only if the subgraphs of $G_{\lambda u + (1-\lambda)v}$
corresponding to the edges with weights $\lambda$ and $1-\lambda$ respectively
are both bi-cliques.
\end{theorem}
\begin{proof}
(\textbf{if part}) Let $u$ and $v$ be the vertices of $\cal R$ and $G_u =  K_V \setminus \{\bigcup_{i=1}^2 K_{S_i}\}$
$G_v = K_V \setminus \{\bigcup_{i=1}^2 K_{R_i}\}$ (since both are bi-cliques). Recall that by definition for any two vertices $u_i$ and $v_j$
in $S_i \cap R_j$ the edge between them is not present. Also the set of edges with weight $\lambda$ 
will be those which are present in $G_u$ and not in $G_v$ vice versa for edges with weights $1-\lambda$. Let 
$H_{\lambda}$ and $H_{1-\lambda}$ be the subgraphs comprising of edges with weights $\lambda$ and $1-\lambda$ respectively. Consider $H_{\lambda}$. Since this graph is given to be a bi-clique we can assume it to 
be $H_{\lambda} = K_{V'} \setminus \{K_{T_1} \cup K_{T_2}\}$. We now choose hyperplanes such that their intersection gives us $G_{\lambda u + (1-\lambda)v}$. For every $i,j,k \in V'$ such that $i,j$ is in some $T_{i_1}$ and $k$ is in some $T_{i_2}$, $i_1,i_2 \in \{1,2\}$ with $T_{i_1} \neq T_{i_2}$ choose the hyperplanes $z_{ij} + z_{jk} = z_{ik}$ and $z_{ij} + z_{ik} = z_{jk}$ among the  set of planes. Note that this implies that $z_{ik} = z_{jk}$ and $z_{ij} = 0$. As a result of selecting these hyperplanes we will get all the variables $z_{ij}$ where $\{i,j\}$ is an edge in $H_{\lambda}$ to be equal. Take this equal value as $\lambda$. Also for all $\{i,j\}$ which are not edge set of $H_{\lambda}$ will have weight 0. Repeat the same exercise of choosing hyperplanes for the subgraph  $H_{1-\lambda}$ but this time instead of taking the equal value as $\lambda$ take the value as $1-\lambda$. For all the rest of the edges $\{i,j\}$ with weight 1 choose the hyperplane $z_{ij} = 1$. To link these values we need to choose some other planes. For all $i,j,k$ such that $\{i,j\} \in H_{\lambda}$ and $\{j,k\} \in H_{1-\lambda}$ choose the plane $z_{ij} + z_{jk} = z_{ik}$. It is now easy to verify that the intersection of all these planes indeed gives the line segment $\lambda u + (1-\lambda) v$.

(\textbf{only if part}) Let $u$ and $v$ be vertices of $\cal R$ and the graph $G_{\lambda u + (1-\lambda) v}$ 
doesn't satisfy above mentioned condition. Similar to Lemma \ref{triangle-edge}, it can be verified that in this case it is always the case that
at least one the following two scenarios will arise: \\
1. The subgraph $H_{\lambda}$ has more than one edge and one can find a pair $\{i,j\}$ such
that $w_{\{i,j\}}$ is $\lambda$ (or $1-\lambda$) and there is no
pair of the form $\{j,k\}$ or $\{i,k'\}$ such that $w_{\{j,k\}} = 0$ or $w_{\{i,k'\}} = 0$. \\
2. The subgraph $H_{\lambda}$ (or $H_{1-\lambda}$) is a collection of disconnected multi-cliques. \\

If it is the first case let $\{i,j\}$ and $\{i_1,j_1\}$ be the two pairs which have weight 
$\lambda$ (w.l.g. assume it is $\lambda$) and since both have value $\lambda$, the intersection 
of the planes which we have chosen implies $z_{ij} = z_{i_1j_1}$. This means there must be some 
be some $k$ such that one of $z_{ij} + z_{jk} = z_{ik}, z_{ij} + z_{ik} = z_{jk} $ and
$z_{ik} + z_{jk} = z_{ij}$ is chosen and $z_{jk} = z_{i_1j_1}$ or $z_{ik} = z_{i_1j_1}$ is implied
by the rest of the planes chosen. But in both these cases, one of $z_{ik}$ or $z_{jk}$ is zero 
which is a contradiction.  \\
For the second case, let $\{i_1,j_1\}$ and $\{i_2,j_2\}$ be two pairs which are in different
multi-cliques but $w_{\{i_1,j_1\}} = w_{\{i_2,j_2\}}$. Therefore, $z_{i_1j_1} = z_{i_2j_2}$ 
must be implied by the chosen set of hyperplanes. But from the discussion presented before, 
such a scenario implies that both $i_1j_1$ and $i_2j_2$ have to be in a connected graph which 
has to be a biclique.
\end{proof}

\subsection{Positive Semi-Definite Constraint}
We now investigate the surface defined by the positive semi-definite constraint $\textbf{1} - Z \succeq 0$.
From Observation 1,  the region defined by the this constraint is 
enclosed inside the region defined by the triangle inequalities. 
In this direction we will prove certain interesting results again relating the 
graphs which some of these matrices correspond to. \\
Given a symmetric $n \times n$ matrix $A$ with $\pm 1$ entries define a new matrix $\widetilde{A}$
such that $\widetilde{A}_{ij} = 1$ if $A_{ij} = -1$ and $\widetilde{A}_{ij} = 0$ if $A_{ij} = 1$.
The matrix $\widetilde{A}$ can be treated as the adjacency matrix of a graph on vertices 
$\{v_1,v_2,\dots v_n\}$. We now prove the following lemma which will be of interest in the 
further discussion.

\begin{lemma}\label{sdp}
Given a symmetric matrix $A= [a_{ij}]$ with $\pm 1$ entries, the expression 
$E(x_1,x_2, \dots, x_n) = \displaystyle \sum_{i}^n a_{ii}x_i^2 + 2 \sum_{i<j}^n a_{ij}x_ix_j$ is non-negative for all
$x_i'$s $\in \mathbb{R}$, iff there exist $b_1,b_2, \dots b_n \in \{1,-1\}$ such that $E$ can be 
expressed as $(b_1x_1 + b_2x_2 +  \dots + b_nx_n)^2$.
\end{lemma}
\begin{proof}
Clearly one way is trivial, i.e. if $E$ is of the above form then it must be non-negative. For the converse 
part we have to show that for all expressions $E$ which are not of this form we can find some values of 
$x_i$'s $i = 1,2, \dots, n$, for which the value of expression these choice of $x_i$'s becomes negative. We will 
denote $\textbf{x} = (x_1,x_2, \dots, x_n)$. \\
Firstly, it is easy to see that any $E$ that is non-negative for all $x_i'$s must have the values of $a_{ii}'$s
as 1 because if any $a_{ii} = -1$ then the expression will be negative for the vector $\textbf{x}$ which is $a$
at the $i^{th}$ position and 0 otherwise where $a$ is a non-zero number. Now we would show that if $E$ is not
of the form $(b_1x_1 + b_2x_2 +  \dots + b_nx_n)^2$ then there always exists a triple $i,j,k$ all three distinct
such that among $a_{ij},a_{jk},a_{ik}$ either all are -1 or two are 1 and one is -1. It is easy to see that under this
assumption we will be done as for both these cases we can find an $\textbf{x}$ such that $E(\textbf{x}) < 0$.
If it is the first case i.e. all are -1's then take $\textbf{x}$ as the vector with $a$ at the positions $i,j,k$
and 0 otherwise. The value of $E$ at this $x$ will be $-3a^2 < 0$. If it is the other case then w.l.g assume that 
$a_{ij}= a_{jk} = 1$ and $a_{ik} = -1$. In this case we can choose $\textbf{x}$ which has $a$ at positions $i$  
and $k$ and $-a$ at position $j$. Again the value of the expression will be $-3a^2 < 0$. \\
We now have to prove that our assumption is always true. We will prove this by induction on $n$ \\
\textbf{Base:} Can easily be verified for $n=4$. \\
\textbf{Induction:} Assuming the above statement holds for $k = n$, we have to show it for $k = n+1$.
The above statement can be interpreted in terms of a graph. Given a matrix $A:= a_{ij}$, consider a 
weighted clique on $n$ vertices in which weight of an edge is 1 or -1. 
Therefore every expression $E$ represents a clique. If it is of the form 
$(b_1x_1 + b_2x_2 +  \dots + b_nx_n)^2$, then we can partition the vertex set of the corresponding graphs into 
two sets $S^{+}$ and $S^{-} = V \setminus S^{+}$ such that weights of all edges in $E(S^{+},S^{-})$ will be -1
and all other edges will be 1. Suppose the statement doesn't hold for $k = n+1$ i.e. there exists an expression
which is not of the form $(b_1x_1 + b_2x_2 +  \dots + b_nx_n)^2$ but no triplet exists which satisfies our condition,
i.e. all triples are either 1 or two are 1 and one is -1. In such a case, remove one vertex from the 
set and this property still holds for all triples and hence by induction we can assume this new graph 
can be decomposed into two sets $S^{+}$ and $S^{-}$ as above. Now, if we put the removed vertex back 
then it is easy to verify that we cannot preserve the initial property. 
\end{proof}

\begin{theorem}\label{sdppm1}
An $n \times n$ symmetric matrix $A$ with $\pm 1$ entries is positive semidefinite if and only if 
the graph corresponding to $\widetilde{A}$ is a complete bipartite graph on vertices $\{v_1,v_2,\dots v_n\}$. 
\end{theorem}
\begin{proof}
If the given matrix $A$ is a positive semi-definite matrix then for all vectors 
$\textbf{x} = (x_1,x_2, \dots, x_n) $, $\textbf{x}^{T}A\textbf{x} \geq$. Now in general
for a symmetric matrix $A$, $\textbf{x}^{T}A\textbf{x}$ can be expanded as 
\begin{center}
$\displaystyle \textbf{x}^{T}A\textbf{x} = \sum_{i} a_{ii}x_i^2 + 2 \sum_{i<j} a_{ij}x_ix_j$ 
\end{center}
In our case, each $a_{ij} = \pm 1$. Now, we can appeal to Lemma \ref{sdp} to conclude that 
the above expression will be non-negative iff it is of the form 
$(b_1x_1 + b_2x_2 +  \dots + b_nx_n)^2$ for some $b_i$'s $\in \{1,-1\}$. From the 
proof of Lemma \ref{sdp} the weighted graph can be partitioned into two sets $S^{+}$
and $S^{-}$ such that edges of $E(S^{+}, S^{-})$ are of weight -1 and rest have weight 1.
As per the definition of $\widetilde{A}$ edges with weight 1 are removed and rest 
have weight 1, which makes the graph corresponding to $\widetilde{A}$, a complete bipartite 
graph.
\end{proof}

\begin{lemma}\label{biclique}
All points $Z$ of the form $\lambda A$ where $A:= [a_{ij}]$ corresponds to a bi-clique, are in $\cal P$
for $\lambda \in [0-2]$.
\end{lemma}
\begin{proof}
$X:= [x_{ij}]$ be the matrix such that $X_{ij} = 1 - (\lambda a_{ij})^2$. Now $\lambda A$ 
will be in $\cal S$ iff the matrix $X$ is positive semi-definite.
Since $A$ corresponds to the adjacency matrix of a bi-clique, there will exists two disjoint 
non-empty subsets of $V = {1,2, \dots, n}$, $S_1$ and $S_2$ such that $S_1 = A \setminus S_2$
and for all $i \in S_1$ and $j \in S_2$, $a_{ij} = 1$. Also for all $i,j \in S_1$, $a_{ij} = 0$. 
and $i,j \in S_2$, $a_{ij} = 0$. Since it is the adjacency matrix of a graph $a_{ii} = 0$. This implies
that $x_{ij} = 1-\lambda$ for all $i \in S_1$ and $j \in S_2$ and $x_{ij} = 1$ for all $i,j \in S_1$
and $i,j \in S_2$. Also $x_{ii} = 1$. Now we know that the matrix $X$ will be a PSD matrix iff 
there exists some $n$, $n$-dimensional vectors $u_1,u_2, \dots, u_n$ such that 
$x_{ij} = \left \langle u_i, u_j \right \rangle$. Since $x_{ii} = 1$ all of these have to be unit vectors.
Since $1-\lambda$ takes values in the range $[-1,1]$. There will always be two vectors $n_1$ and $n_2$
on the unit sphere such that $\left \langle n_1,n_2 \right \rangle = 1-\lambda$. 
Therefore, we can choose the vectors $u_1, \dots, u_n$  as for all $i \in S_1$ take $u_i = n_1$
and for all $i \in S_2$ take $u_i = n_2$. It is easy to verify that for all values of $i,j$,
$x_{ij}$ is indeed  $\left \langle u_i, u_j \right \rangle$.
\end{proof}

The following also is provable similar to Lemma \ref{biclique}

\begin{lemma}
Let $B$ be a partial clique from is obtained by removing 
$k$ cliques form $K_n$, then all vectors corresponding to $\lambda B$ belong to $\cal P$
for $\lambda \in [0- \lambda_k]$ where $\lambda_k \in [0-2]$.
\end{lemma}

\begin{lemma}
An edge $e$ of the cube $[0-2]^d$ is completely contained inside $\cal P$ or completely outside it.
\end{lemma}

\section{Optimizing Over the Vertices is Easy}
In this section we give a characterization of the vertices and arcs of $\cal F$ and show how
the objective function can easily be optimized over the vertices. Let $\Gamma$ be the 
the points of the intersection of the hyperplane supporting {\cal H} (denoted by ${\cal H}^{*}$) the segments of the form $\lambda B$
where $B$ is a partial clique. The vertices $\Delta$, of  $\cal F$, can be divided into types, 

\begin{itemize}
\item[\textbf{Type 1}:] Points formed by the 
intersection of of  ${\cal H}^{*}$ with the edges of $\cal P$ which we call $\Delta_1$ and 
\item[\textbf{Type 2}:] Vertices of the hypercube which satisfy $\cal H$ and $\cal P$. Let this set be denoted by $\Delta_2$. 
\end{itemize}
Notice that $\Delta_1 \subset \Gamma$.

We now prove an interesting result showing that we can infact optimize this
objective function very efficiently over the set of vertices. The reason is that once we find
out the intersection points of ${\cal H}^{*}$ with the with the edges of $\cal T$ 
and subsequently find the points among these which minimizes the 
objective function $\sum_{i,j \in G} \sqrt{z_{ij}}$. But things become simpler
because there aren't many edges which intersect with the levels of the objective function. 

More specifically, our result essentially answers the following question:
given a connected graph $G = (V,E)$ is it possible to efficiently find a partial clique $G'$ on the 
vertex set $V$ such that if an edge $\{i,j\}$ is present in $G$ then it is also present
in $G'$ and if it is not present in $G$ then it is also not present in $G'$. 
The following theorem says given a connected graph $G = (V,E)$ and a subset of 
edges $E'$ one can easily decide whether there exists 
a partial clique $G^* = (V, E^*)$ such that $E' \subseteq E^{*}$ and 
$E \setminus E' \subseteq \bar{E^{*}}$, where $\bar{E^{*}}$ is the set of edges in the complement
of $G^*$. Before that we state the following observation that directly follows from Theorem \ref{partialclique}.

\begin{observation}\label{component}
If $G = (V,E)$ is a partial-clique on $n$ vertices that does not contain a fixed set of 
edges $E'$ which forms a connected component then it does not contain the 
clique defined by the vertices induced over $E'$. 
\end{observation}

\begin{theorem}\label{vertexeasy}
Given a connected graph $G = (V,E)$ and a subset of edges $E'$ there exists a unique (if any)
partial clique $G^* = (V, E^*)$ such that $E' \subseteq E^{*}$ and 
$E \setminus E' \subseteq \bar{E^{*}}$, where $\bar{E^{*}}$ is the set of edges in the complement
of $G^*$ and that partial-clique can be found efficiently. 
\end{theorem} 
\begin{proof}
If such a partial-clique exists then it will be of the form $K_V \setminus \{\bigcup_{i=0}^r K_{S_i} \}$
for some subsets $S_1, S_2, \dots, S_r \subset V$ for some $r$. Let $C_1, C_2,\dots, C_k$ be the 
connected components of $G' = (V, E \setminus E')$ with the corresponding vertex sets as $V_1, V_2, \dots, V_k$.
Clearly all these vertex sets are pairwise disjoint. Using Observation \ref{component} 
we can infer that each $K_{V_i}$ is not present in the graph. Now we have to show that some other 
clique or a clique that contains some of these cliques is not missing. The first possibility is easily 
ruled out as $G$ is a connected graph and hence any other clique will contain at least one edge 
in $E$ which will violate the condition that $E' \subseteq E^{*}$. Similarly, for the other case as well
if some other $K_{V'}$ is removed such that $V_s \subset V'$ for some $s$, then also the same condition
will be violated. As evident from the proof such a partial-clique (if exists) can be computed efficiently.
\end{proof}

\section{Conclusion}
In this paper, we propose a well-structured family of programs
 called concave programming and investigate the combinatorial geometric structure
of the feasible set of the program and show how to possibly use them in the context of 
graph partitioning problems like {\sc c-Balanced Separator}. This is a major paradigmatic shift
to attack these problems. It will of immense use to see whether or not such techniques
can give us improved approximation factor for other problems. This
also gives us hope that for many of the problems for which optimal approximation factors
are not known one can possibly rely upon some ``nice'' programs 
which are although not convex but can be potential candidates for 
polynomial time solvability because of their geometric structure. Since this family is a new form of mathematical programming 
that is being used in an approximation algorithm, progress both in the direction of hardness and algorithms 
will provide more insights into the nature of these concave programs and can potentially lead us to 
optimal inapproximability results for various graph-partitioning problems
Another tempting direction inspired from the recent results on the {\sc Unique Games} \cite{ABS2010}
is to exploit these combinatorial geometric ideas to design sub-exponential time $O(1)$-approximation algorithms  
for the problem. 

\section{Acknowledgments}
The author would like to thank Sanjeev Arora for discussing the prospects of mathematical programming
paradigms beyond SDPs. Thanks to Purushottam Kar for going through an earlier draft of the paper
and sending his comments.

\appendix
\section{Appendix}

\subsection{Proof of Theorem \ref{concave}}
\begin{proof}
Since $z^{p/2}$ is concave for $p<2$ for $z>0$, and the sum of concave functions
is also concave, the objective function is clearly concave. For the constraints
defining the feasible set, $\sum_{i < j} z_{ij} \geq c(1-c)n^2$ and $z_{ii} = 0$
are convex. The constraint $\textbf{1} - Z \succeq 0 $ can be shown to be convex 
as follows: Let $Z_1$ and $Z_2$ be two matrices corresponding to the variables $z_{ij}$'s
which lie in the feasible set. Therefore, they satisfy $\textbf{1} - Z_1 \succeq 0$ and 
$\textbf{1} - Z_2 \succeq 0$. Now, consider the line segment for $\lambda \in [0-1]$
$\lambda Z_1 + (1-\lambda) Z_2$ and the matrix $\textbf{1}-(\lambda Z_1 + (1-\lambda) Z_2)$. This
is positive semidefinite as it can be rewritten as $\lambda(\textbf{1}- Z_1) + (1-\lambda)(\textbf{1} - Z_2)$ 
which is a sum of two PSD matrices. \\
The only type of constraint left are the triangle inequality constraints. Consider 
an inequality of this type say $z_{ij}^{p/2} + z_{jk}^{p/2} \geq z_{ik}^{p/2}$.  
In general, let us look at the region $x^{r} + y^{r} \geq z^r$ for $0 < r < 1$.
If $r = 1/q$ for $q > 1$ then this region is same as $ \left(x^{1/q} + y^{1/q}\right)^{q} \geq z$.
Let $p_1 = (x_1,y_1,z_1)$ and $p_2 = (x_2,y_2,z_2)$ be two points which lie in this region, i.e. 
$\left({x_1}^{1/q} + {y_1}^{1/q}\right)^{q} \geq z_1$ and 
$\left({x_2}^{1/q} + {y_2}^{1/q}\right)^{q} \geq z_2$. To prove 
the convexity of the region we need to show that for any $\lambda \in [0-1]$,
$(\lambda x_1 + (1-\lambda) x_2, \lambda y_1 + (1-\lambda) y_2, \lambda z_1 + (1-\lambda) z_2)$
also lies inside the region for all such points $p_1$ and $p_2$. Therefore, we have to show
$\displaystyle \lambda z_1 + (1-\lambda) z_2 \leq \left((\lambda x_1 + (1-\lambda) x_2)^{\frac{1}{q}} + (\lambda y_1 + (1-\lambda) y_2)^\frac{1}{q} \right)^{q}$.
Thus we will be done if we show 
$\displaystyle \lambda ({x_1}^{1/q} + {y_1}^{1/q})^{q} + (1-\lambda) ({x_2}^{1/q} + {y_2}^{1/q})^{q} \leq \left((\lambda x_1 + (1-\lambda) x_2)^{\frac{1}{q}} + (\lambda y_1 + (1-\lambda) y_2)^\frac{1}{q} \right)^{q}$.
which is equivalent to proving that the function $f(x,y) = \left({x}^{\frac{1}{q}} + {y}^{\frac{1}{q}}\right)^{q}$
is concave. We will prove this by showing that the Hessian of this function is negative-definite for all $x,y$. We now
compute the entries of the Hessian matrix. The following calculations are easy to verify,
\begin{eqnarray*}
\frac{\partial f}{\partial x} &=& \left(1 + \frac{y^{\frac{1}{q}}}{x^{\frac{1}{q}}} \right)^{q-1} ; 
\frac{\partial f}{\partial y} = \left(1 + \frac{x^{\frac{1}{q}}}{y^{\frac{1}{q}}} \right)^{q-1} ; 
\frac{\partial^2 f}{\partial x^2} = -\left(\frac{q-1}{q}\right) \left(1 + \frac{x^{\frac{1}{q}}}{y^{\frac{1}{q}}} \right)^{q-2} \frac{y^{\frac{1}{q}}}{x^{\frac{q+1}{q}}} ;\\
\frac{\partial^2 f}{\partial y^2} &=& -\left(\frac{q-1}{q}\right) \left(1 + \frac{y^{\frac{1}{q}}}{x^{\frac{1}{q}}} \right)^{q-2}  \frac{x^{\frac{1}{q}}}{y^{\frac{q+1}{q}}} ;  \frac{\partial^2 f}{\partial x \partial y} = \left(\frac{q-1}{q}\right) \left( \frac{1}{x^{\frac{1}{q}}} + \frac{1}{y^{\frac{1}{q}}}\right)^{q-2} \frac{1}{y^{\frac{1}{q}}x^{\frac{1}{q}}} = \frac{\partial^2 f}{\partial y \partial x}
\end{eqnarray*}

In order to show that the Hessian is negative-definite we have to show that for any $\alpha, \beta \in \mathbb{R}$,
the following expression is always non-positive for all $x,y > 0$ (for $x,y$ as 0 the derivatives do not exist):
\begin{eqnarray*}
&& \alpha^2 \frac{\partial^2 f}{\partial x^2} + \beta^2 \frac{\partial^2 f}{\partial y^2} +  2 \alpha \beta \frac{\partial^2 f}{\partial x \partial y} \\
& = & -\left( \frac{q-1}{q} \right) \left[ \alpha^2 \left(1 + \frac{x^{\frac{1}{q}}}{y^{\frac{1}{q}}} \right)^{q-2}  \frac{y^{\frac{1}{q}}}{x^{\frac{q+1}{q}}} + \beta ^2 \left(1 + \frac{y^{\frac{1}{q}}}{x^{\frac{1}{q}}} \right)^{q-2} \frac{x^{\frac{1}{q}}}{y^{\frac{q+1}{q}}} - \left( \frac{1}{x^{\frac{1}{q}}} + \frac{1}{y^{\frac{1}{q}}}\right)^{q-2} \cdot \frac{2\alpha \beta}{y^{\frac{1}{q}}x^{\frac{1}{q}}}\right] \\
&=& -\left( \frac{q-1}{q} \right) \left({x}^{\frac{1}{q}} + {y}^{\frac{1}{q}}\right)^{q} \left[ \frac{\alpha^2y^{\frac{1}{p}}}{x^{\frac{2q-1}{q}}} + \frac{\beta^2x^{\frac{1}{q}}}{y^{\frac{2q-1}{q}}} - \frac{2\alpha \beta}{x^{\frac{q-1}{q}}y^{\frac{q-1}{q}}} \right] \\
& = & -\left( \frac{q-1}{q} \right) \left({x}^{\frac{1}{q}} + {y}^{\frac{1}{q}}\right)^{q} \left[ 
\frac{\alpha^2y^2 + \beta^2x^2 - 2\alpha \beta xy }{x^{\frac{2q-1}{q}}y^{\frac{2q-1}{q}}} \right] 
= -\left( \frac{q-1}{q} \right) \left({x}^{\frac{1}{q}} + {y}^{\frac{1}{q}}\right)^{q} \left[ 
\frac{(\alpha y - \beta x)^2 }{x^{\frac{2q-1}{q}}y^{\frac{2q-1}{q}}} \right] 
\end{eqnarray*}
which is non-positive for all $\alpha, \beta$
This proves that the region $x^{p/2} + y^{p/2} \geq z^{p/2}$ is a convex set for all $0< p < 2$. Hence the 
intersection of all the triangle inequality constraints is also convex.
\end{proof}

\end{document}